\documentclass[11pt]{llncs}
\pdfoutput=1
\usepackage{graphicx}
\usepackage{cite}

\usepackage{times}
\usepackage{mathptm}
\usepackage{amssymb}

\def\PNP{P${}={}$NP}

\makeatletter
\def\hb@xt@{\hbox to }
\makeatother

\let\oldendproof\endproof
\def\endproof{\qed\oldendproof}

\begin{document}

\title{Paired Approximation Problems and Incompatible Inapproximabilities} 

\author{David Eppstein}

\institute{Computer Science Department\\
School of Information \& Computer Science\\
University of California, Irvine\\
\email{eppstein@uci.edu}}

\maketitle   

\begin{abstract}
This paper considers pairs of optimization problems that are defined from a single input and for which it is desired to find a good approximation to either one of the problems. In many instances, it is possible to efficiently find an approximation of this type that is better than known inapproximability lower bounds for either of the two individual optimization problems forming the pair. In particular, we find either a $(1+\epsilon)$-approximation to $(1,2)$-TSP or a $1/\epsilon$-approximation to maximum independent set, from a given graph, in linear time. We show a similar paired approximation result for finding either a coloring or a long path. However,  no such tradeoff exists in some other cases: for set cover and hitting set problems defined from a single set family, and for clique and independent set problems on the same graph, it is not possible to find an approximation when both problems are combined that is better than the best approximation for either problem on its own.
\end{abstract}

\section{Introduction}
Many classical optimization problems are now known to be hard to approximate, among them
the maximum independent set of a graph and the $(1,2)$-Traveling Salesman Problem (specified by an unweighted undirected graph, with adjacent vertices at distance one from each other and nonadjacent vertices at distance two): unless \PNP, the $(1,2)$-TSP cannot be approximated within a factor better than $1+\frac{1}{740}$~\cite{EngKar-ICALP-01} and the maximum independent set cannot be approximated within a factor better than $n^{1-\epsilon}$ for any constant $\epsilon>0$~\cite{Has-AM-99,Zuc-STOC-06}. However, what if both problems are specified by the same input graph, and we would be happy to find a good approximate solution to either one? As we show, this \emph{paired approximation problem} may be efficiently approximated more accurately than either problem may be on its own: a simple linear time approximation algorithm finds for any $\epsilon>0$ either a $(1+\epsilon)$-approximate solution to the $(1,2)$-TSP or a $1/\epsilon$-approximate solution to the maximum independent set.

This result may seem a mere curiosity: in what application would we be satisfied by either one of two such seemingly unrelated outputs? However, this precise problem came up in recent work of the author with Cabello and Klav{\v z}ar on isometric embedding of graphs~\cite{CabEppKla}: using a subroutine that can be interpreted as finding either a short $(1,2)$-TSP or a large independent set of an auxiliary graph derived from the input, we developed a polynomial-time approximation scheme for isometrically embedding simplex graphs into Fibonacci cubes of low dimension. In this problem, an approximate $(1,2)$-TSP may be used directly to construct an embedding, while a large independent set may be used to show that the derived graph has size logarithmic in the input size, allowing the use of exponential-time TSP algorithms.

Beyond this application, the general phenomenon of inapproximability results losing their strength when combined in pairs provides an interesting test on the power of inapproximability theory.
To this end, we detail additional instances in which two classical and difficult-to-approximate optimization problems are defined from a common input. We show that many pairs of inapproximability proofs are \emph{incompatible}: the hard instances for one of the two paired problems are disjoint from the hard instances for the other, so that an approximation algorithm may find a solution to one or the other problem that is better than the known inapproximability bounds for either problem. For other pairs of problems, however, the paired problem is as hard to approximate as the individual problems from which it is formed.

We take as a model for our results the \emph{Hadwiger conjecture} that every graph that requires $k$ or more colors has a $k$-vertex clique minor~\cite{Had-VNgZ-43}. Although much work on this conjecture treats $k$ as constant, we consider $k$ to be a variable, and we interpret the conjecture as a tradeoff between two graph optimization problems: when a graph is worse as an instance of a coloring problem (that is, it requires more colors) it becomes better as an instance of a clique-finding problem (it has larger clique minors), and vice versa. However, for the purposes of approximation algorithms Hadwiger's conjecture is problematic because it is unproven and because known algorithms for finding either a coloring or a clique minor for a given $k$ have a running time that depends badly on~$k$~\cite{KawRee-STOC-09}. For the paired approximation problems that we study we do not need to know the tradeoff in solution sizes at the level of precision provided by Hadwiger's conjecture: it suffices to prove a weaker tradeoff that can be implemented by fully polynomial-time algorithms.  

Our algorithms are all simple and easily implementable; our intent is less to come up with sophisticated new approximation techniques, and more to explore the implications and limitations of the single-problem approximation ratio used as a standard basis for the theoretical analysis of heuristics for hard problems.

\subsection{New Results}

By considering the number of leaves in a depth-first search tree of an input graph, we prove:

\medskip\noindent
\textbf{Independent set and $(1,2)$-TSP.}
Let $\alpha(G)$ denote the independence number of graph $G$, and let $L_{1,2}(G)$ denote the length of the optimal solution to the $(1,2)$-TSP problem defined from $G$. Then we show that $L_{1,2}(G)-\alpha(G)\le n$. The result is constructive: in linear time we can find a cycle with length $L$, and an independent set with size $I$, such that $L-I=n$. As a consequence, it is possible, given a graph $G$ and a parameter $\epsilon>0$, to find in linear time either a $(1+\epsilon)$-approximation to the $(1,2)$-TSP or a $1/\epsilon$-approximation to the independent set. The same idea provides paired approximations for the minimum leaf spanning tree~\cite{SalWie-IPL-07} and independent sets, or for independent sets and the minimum number of paths in a path cover.

\medskip\noindent
By considering the height of a depth-first search tree of an input graph, we prove:

\medskip\noindent
\textbf{Coloring and longest path.}
If $\chi$ is the chromatic number of a graph, and $\ell(G)$ is the number of vertices in its longest path, then $\ell(G)\ge\chi(G)$. In linear time we can find a coloring with $c\ge\chi(G)$ colors, and a $c$-vertex path in $G$. Thus it is possible, given a graph $G$ and a parameter $0<\epsilon<1$, to find in linear time either an $n^{\epsilon}$-approximate longest path or an $n^{1-\epsilon}$-approximate coloring.

\medskip
One may combine height and number of leaves in a single tradeoff, using the observation that the product of the height and number of leaves in any tree is less than $n$. For undirected graphs, this leads only to weak tradeoffs between independent sets and longest paths. However, for directed graphs, we find a paired approximation between the directed longest path problem and the problem of finding the largest vertex set that induces an acyclic subgraph (a complementary problem to the more well-studied feedback vertex set).
Another paired approximation result related to Hadwiger's conjecture follows easily from known techniques:

\medskip\noindent
\textbf{Coloring and clique minor.}
The original tradeoff between graph coloring and clique minors is Hadwiger's conjecture: in any graph with chromatic number $\chi$ there exists a $\chi$-vertex clique minor. Although this remains unproven, standard greedy coloring methods and known algorithms for finding cliques in dense graphs can be combined to result in an algorithm that finds in polynomial time a coloring with chromatic number $c\ge\chi$, and a clique minor with $\Omega(c/\sqrt{\log n})$ vertices. As a consequence, we show that it is possible, given a graph $G$ and a parameter $0<\epsilon<\frac12$, to find in polynomial time either an $O(n^{\epsilon}\log^{1/4} n)$-approximate solution to the clique minor problem or an $O(n^{1-\epsilon}\log^{1/4} n)$-approximate solution to the graph coloring problem. (For $\epsilon\ge \frac12$, this is not better than known approximations for the clique minor problem alone~\cite{AloLinWah-TCS-07}.)

\medskip
We also prove the following negative results.

\medskip\noindent
\textbf{Set cover and hitting set.}
The set cover and hitting set problems are both defined on a family of sets; in the set cover problem, the task is to find a subfamily of as few sets as possible with the same union, while in the hitting set problem the task is to find a set of elements that has a nonempty intersection with every set in the family. With the same complexity assumption as Feige~\cite{Fei-JACM-98}, for all $\epsilon>0$, it is impossible for a polynomial-time algorithm to approximate the paired approximation problem with an approximation ratio better than $(1-\epsilon)\ln n$, where $n$ is the total size of all the sets in the instance. The result is proved by an approximation-preserving reduction from set cover. Our reduction can be made to work for any family of hard instances of the set cover problem, but in order to achieve the best constant factor in the new inapproximability result it is necessary to depend on some details of Feige's previous reduction.

\medskip\noindent
\textbf{Clique and independent set.}
Intuitively, one might expect a good paired approximation for the maximum clique and maximum independent set problems in the same graph, as a dense graph will have no large independent set while a sparse graph will have no large clique. However, again we prove a negative result: unless \PNP, for any $\epsilon>0$,  it is impossible for a polynomial-time algorithm to approximate the paired approximation problem with an approximation ratio better than $n^{1-\epsilon}$. Our proof combines ideas from the previous clique inapproximability proofs with a deterministic construction of a Ramsey graph that avoids both large independent sets and large bicliques; this construction may be of independent interest.

\medskip\noindent
\textbf{TSP and MaxTSP.}
We show for an explicit constant $\epsilon>0$ that, unless \PNP, it is impossible to approximate the paired problem of the TSP and the Maximum TSP within a factor better than $1+\epsilon$. The reduction uses $(1,2)$ metrics, and works as well for the hardness of a paired problem in which two $(1,2)$-TSP instances are determined by a graph and its complement.

\subsection{Related Work}

Many standard results relate two hard optimization problems on the same input: for instance, in any graph, the sizes of a minimum vertex cover and of a maximum independent set sum to the number of vertices. However, for the paired approximation problems we study, this is not the right type of tradeoff: it says that if an instance has a better maximum independent set, it also has a better vertex cover. Our results need tradeoffs in which an improvement in one problem is always balanced against a disimprovement in the other.

Beyond Hadwiger's conjecture itself and the large body of research surrounding it, we are unaware of much past work that could be recognized as a paired approximation algorithm of the type we study here. One exception is the work of Boppana and Halld\'orsson on cliques and independent sets~\cite{BopHal-BIT-92}: they describe a polynomial-time algorithm that finds a clique and an independent set the product of the sizes of which is $\Omega(\log^2 n)$, allowing for a weak tradeoff in approximation quality between the two problems (weaker than the single-problem approximation ratio that they prove, based on this result, in the same paper).
Another work that can be interpreted as solving a paired optimization problem is an algorithm of Bodlaender~\cite{Bod-Algs-93} that uses depth-first search to find either a long path in a graph or a low-width tree-decomposition; if a tree-decomposition is found, dynamic programming can then be used to find long paths. Our algorithms use a depth-first-search based approach inspired by this idea.

In the remainder of this section, we summarize briefly the known inapproximability bounds (based on the theory of probabilistically checkable proofs~\cite{AroSaf-JACM-98}) for the problems we study.

\medskip\noindent
\textbf{$(1,2)$-TSP}. The $(1,2)$-TSP problem is a special case of the traveling salesman problem in which all pairwise distances are either one or two. An input for this problem may be given as an undirected graph, with pairs of vertices at distance one represented by edges and pairs at distance two represented by non-edges; the task is to find a cyclic ordering of the vertices minimizing the sum of distances between adjacent vertices in the ordering. A nearly-equivalent problem (differing only when the input graph is Hamiltonian) is to find a collection of disjoint paths, covering all the vertices of the graph and minimizing $n$ plus the total number of paths. The $(1,2)$-TSP was the version of the TSP used to prove its NP-completeness by Garey and Johnson~\cite{GarJoh-79}, and one of the problems mentioned as being MAXSNP-hard in the paper~\cite{PapYan-JCSS-91} originally defining the class MAXSNP of approximation problems (a class now known to be approximable to within a constant factor but not better~\cite{AroLunMot-JACM-98}). Although some classes of graphs define easy-to-approximate TSP instances~\cite{Kle-FOCS-05}, it is impossible to approximate $(1,2)$-TSP for general graphs in polynomial time with an approximation ratio better than $1+\frac{1}{740}$, unless \PNP~\cite{EngKar-ICALP-01}. The best known polynomial approximation for this problem has approximation ratio $8/7$~\cite{BerKar-SODA-06}, improving a previous $7/6$ bound of Papadimitriou and Yannakakis~\cite{PapYan-MOR-93}.

\medskip\noindent
\textbf{Clique and independent set}. An independent set in an undirected graph is a subset of the vertices such that no edge has both endpoints in the subset. Finding a maximum independent set was one of Karp's original 21 NP-complete problems~\cite{Kar-CCC-72}.  For optimization and approximation purposes it is essentially the same as the problem of finding the maximum clique in a complementary graph. The best known polynomial-time approximation ratio for this problem is $O(n(\log\log n)^2/\log^3 n)$~\cite{Fei-SJDM-04}. Such near-linear approximation ratios are the best possible: unless \PNP, the maximum independent set cannot be approximated within a factor better than $n^{1-\epsilon}$ for any constant $\epsilon>0$~\cite{Has-AM-99,Kho-FOCS-01,Zuc-STOC-06}.

\medskip\noindent
\textbf{Coloring}. Graph coloring is another of Karp's 21 problems. In an early work on approximation algorithms, Johnson~\cite{Joh-SEC-74} showed that a greedy algorithm can be used to approximate the chromatic number of a graph with an approximation ratio of $O(n/\log n)$. Although better approximations are known when the chromatic number is small~\cite{KarMotSud-JACM-98}, the current best upper bound on the approximation ratio of a polynomial-time graph coloring algorithm for arbitrary graphs is not much better, $O(n/\log^2 n)$~\cite{BopHal-BIT-92}. Unless \PNP, graph coloring is hard to approximate with an approximation ratio better than $n^{1-\epsilon}$, for any $\epsilon>0$~\cite{FeiKil-CC-96,Kho-FOCS-01,Zuc-STOC-06}. 

\medskip\noindent
\textbf{Clique minor}. Long studied from the graph theoretic point of view, the size of the largest clique in a graph (its Hadwiger number) has not attracted as much attention from the point of view of computational complexity as the other problems described here. It was not until recently that it was even proven NP-complete~\cite{Epp-09}. Alon et al.~\cite{AloLinWah-TCS-07} observe that an algorithm of Kostochka~\cite{Kos-Comb-84} for finding a clique minor whose size is related to the density of the given graph may be used as an approximation algorithm, achieving an approximation ratio of $O(\sqrt n)$. On the lower bound side, Wahlen~\cite{Wah-TCS-09} shows that, unless \PNP, there can be no polynomial-time approximation scheme for the problem.

\medskip\noindent
\textbf{Longest paths}. Longest paths in directed or undirected graphs have been much studied, but as Bj\"orklund et al.~\cite{BjoHusKha-ICALP-04} write, ``this problem is notorious for the difficulty of understanding its approximation hardness.'' Undirected longest paths are NP-hard based on an easy reduction from Hamiltonian paths~\cite{GarJoh-79}, and the directed case is impossible to approximate to within a factor smaller than $n^{1-\epsilon}$, for any $\epsilon>0$, unless \PNP~\cite{BjoHusKha-ICALP-04}. With stronger assumptions the directed problem is hard to approximate to within a factor of $n/\log^{2+\epsilon} n$ for any $\epsilon>0$~\cite{BjoHusKha-ICALP-04}. Color coding~\cite{AloYusZwi-JACM-95} provides an approximation for the directed longest paths with approximation ratio $O(n/\log n)$. The known inapproximability results for undirected longest paths are relatively weak~\cite{KarMotRam-Algo-97}, but the best known polynomial-time approximation algorithm for this case achieves an approximation ratio of only $O(n(\log\log n/\log n)^2)$~\cite{BjoHus-SJC-03}, so the undirected problem may be similarly difficult to approximate.

\medskip\noindent
\textbf{Set cover and hitting set}. Set cover (another of Karp's 21 problems) and the hitting set problem are superficially different (one asks for a subfamily of sets, the other a single set) but both can be described in the same way. One may represent a set family using an incidence matrix, a 0-1 matrix with the rows indexed by elements, the columns indexed by sets, and a 1 in every entry that corresponds to an element and a set containing that element; then set cover is the problem of finding a small set of columns the sum of which is entirely nonzero, while hitting set is the problem of finding a small set of rows the sum of which is entirely nonzero. Transposing the matrix produces a dual instance in which the two problems exchange roles, so they are the same for purposes of approximation. A greedy algorithm that at each step chooses a set that covers the largest possible number of remaining uncovered elements will find an approximation to an $n$-element set cover problem with approximation ratio $\ln n-\ln\ln n+\Theta(1)$~\cite{Chv-MOR-79,Joh-JCSS-74,Lov-DM-75,Sla-Algs-97}. Feige~\cite{Fei-JACM-98} proves (with a complexity-theoretic assumption) that, for any $\epsilon>0$, no polynomial algorithm can approximate set cover better than $(1-\epsilon)\ln n$. The set family constructed by Feige's reduction is sparse, in the sense that each element belongs to $o(n)$ sets of the family, so the same $(1-\epsilon)\ln n$ inapproximability bound holds also when $n$ measures the sum of the sizes of the sets in the set family rather than the number of elements.

\section{Leaves of the depth-first search forest}

Suppose that we are given as input a graph in which we would like to find either a small $(1,2)$-TSP tour or a large independent set. As we now show, both problems can be approximated using depth-first search. The solutions depend oppositely on the number of leaves of the DFS forest: more leaves lead to worse tours and better independent sets. This tradeoff between the two problems leads to our approximation bounds.

\begin{figure}[t]
\centering\includegraphics[width=4in]{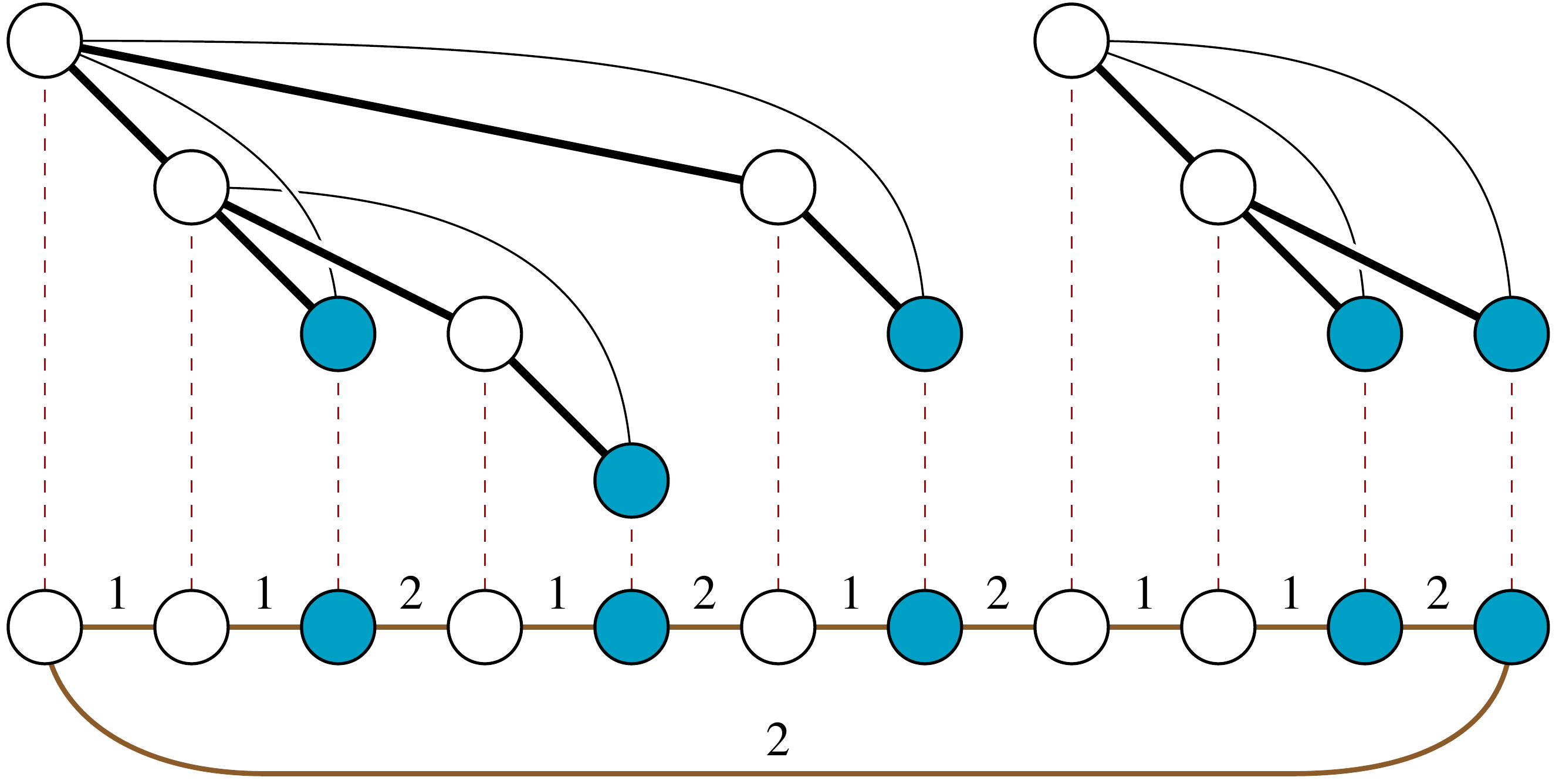}
\caption{Preorder traversal of a DFS forest forms a $(1,2)$-TSP tour in which the first vertices of the length-2 steps form an independent set of DFS leaves.}
\label{fig:dfs-tour}
\end{figure}

\begin{theorem}
In linear time, given an $n$-vertex graph $G$, it is possible to find a tour of length $L$ through the vertices of $G$ and an independent set of $I$ vertices in $G$, such that $L-I\le n$.
\end{theorem}

\begin{proof}
Perform a depth-first search in $G$, let the tour be a preorder traversal of the depth-first search forest, and let the independent set be the set of leaves of the depth-first search forest (Figure~\ref{fig:dfs-tour}). A depth-first search forest has the property that the leaves form an independent set: it is not possible for both endpoints of an edge $uv$ of $G$ to be leaves, because whichever of $u$ and $v$ is visited first by the DFS must be the root of a subtree containing the other of the two vertices. The preorder traversal follows edges of the depth-first search tree except after each leaf, so the number of length-2 steps in the tour is at most $I$ and the result follows.
\end{proof}

A little more strongly, the equality $L-I=n$ holds except in the case that $L=n$ and $I=1$. An alternative proof, using a technique from \cite{CabEppKla}, is to consider the edges of the graph in an arbitrary order, adding each edge to a greedily constructed collection of paths if it does not complete a non-Hamiltonian cycle or form a claw, and then to form an independent set by choosing one endpoint per path.

\begin{corollary}
In any $n$-vertex graph, if $L^*$ and $I^*$ denote the qualities of the optimal solution to the $(1,2)$-TSP and the independent set problem defined from the graph, then $L^*-I^*\le n$.
\end{corollary}

This tradeoff is tight: for any $I>0$ and any $n\ge I$ one can find an $n$-vertex graph with independence number $I$ and with optimal tour length $n+I$, by choosing the graph to be a disjoint union of $I$ cliques.

\begin{theorem}
Given a graph $G$ with $n$ vertices and $m$ edges, and an input parameter $\epsilon>0$, it is possible to find in time $O(n+m)$ either a tour of length $L$ through the vertices of $G$ that approximates the optimal tour to within a factor of $1+\epsilon$, or an independent set in $G$ that approximates the optimal independent set to within a factor of $1/\epsilon$.
\end{theorem}

\begin{proof}
Compute the tour and independent set as described above. If $L\le n(1+\epsilon)$, then return the tour. Otherwise, $I\ge\epsilon n$; in this case, return the independent set.
\end{proof}

The same technique produces a cover of the graph by $I$ paths, and a spanning tree of the graph with $I$ leaves. An algorithm that chooses to return an independent set if $I\ge n^{1-\epsilon}$ and to return a path cover or spanning tree otherwise will achieve an approximation ratio of $n^\epsilon$ for the independent set problem or an approximation ratio of $n^{1-\epsilon}$ for the path cover or minimum-leaf spanning tree problems.

\section{Height of the depth-first forest}

\begin{figure}[t]
\centering\includegraphics[width=2.5in]{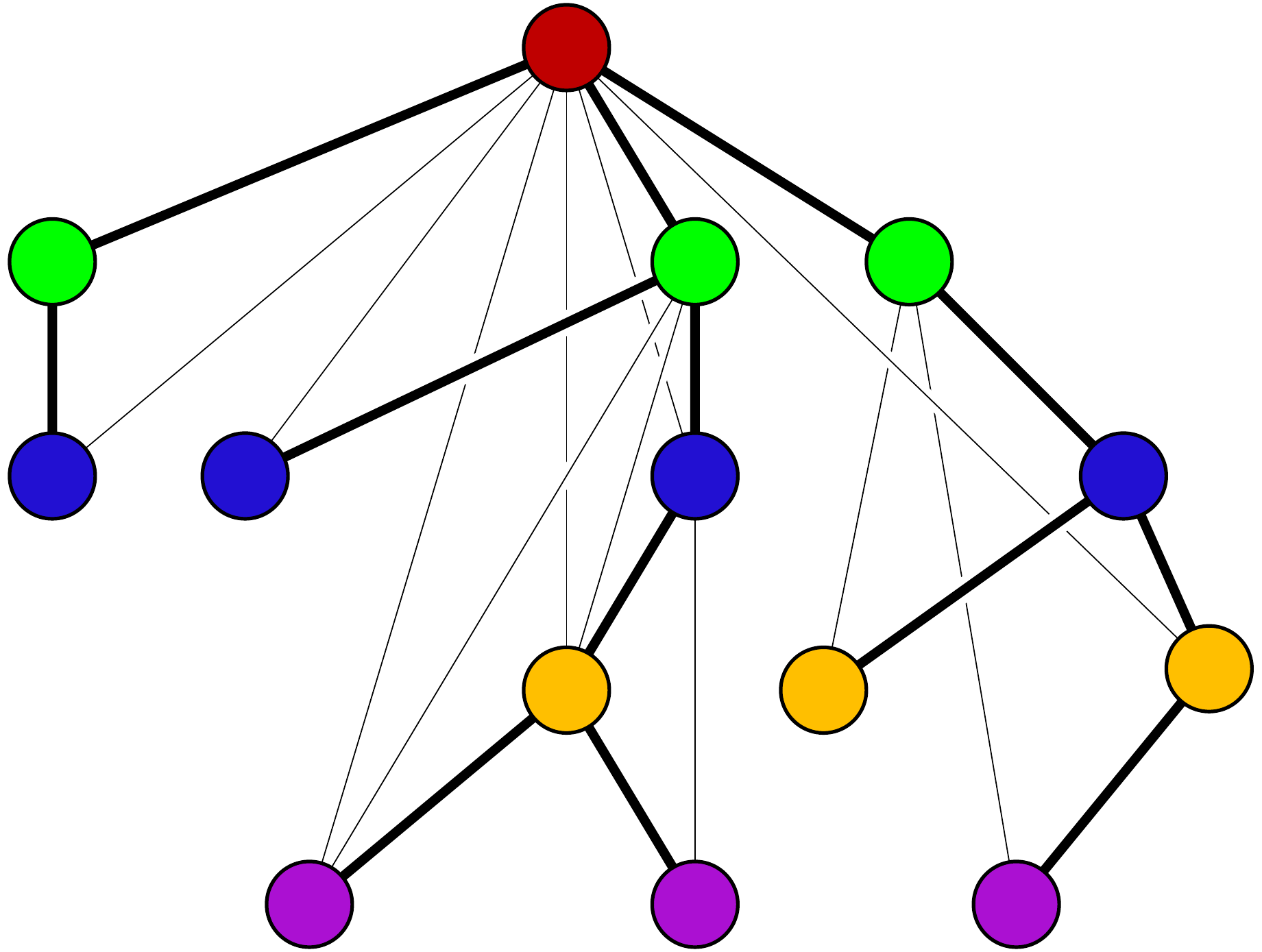}
\caption{Partitioning the vertices of a graph according to the levels of its DFS forest produces a proper coloring.}
\label{fig:dfs-color}
\end{figure}

\begin{theorem}
In linear time, given an $n$-vertex graph $G$, it is possible to find a coloring of the graph and a path in the graph such that each color class of the coloring contains exactly one path vertex.
\end{theorem}

\begin{proof}
Perform a depth-first search in $G$, and for each vertex $v$ at distance $i$ from the root of its tree assign the $i$th color to $v$ (Figure~\ref{fig:dfs-color}). Let the path be a longest root-to-leaf path in the depth-first search forest.
\end{proof}

\begin{corollary}
In any $n$-vertex graph $G$, if $P^*$ and $\chi$ denote the number of vertices in the longest path and the chromatic number of $G$ respectively, then $\chi\le P*$.
\end{corollary}

The tradeoff is tight: for a disjoint union of cliques, all of which have at most $\chi$ vertices and one of which has exactly $\chi$ vertices, $\chi=P^*$. An alternative proof for the corollary uses a greedy coloring algorithm in which the minimum-degree vertex is removed, the rest of the graph colored recursively, and the removed vertex colored with the minimum available color: if this uses $k$~colors, then some subgraph has minimum degree $k-1$ and a greedy algorithm can find a path of $k$ or more vertices in that subgraph.

\begin{theorem}
In linear time, given an $n$-vertex graph $G$ and an input parameter $0<\epsilon<1$, it is possible to find either a path that approximates the longest path to within a factor of $n^{\epsilon}$, or a coloring that approximates the chromatic number to within a factor of $n^{1-\epsilon}$.
\end{theorem}

\begin{proof}
Compute the coloring and path as described above. If the path contains at least $n^{1-\epsilon}$ vertices, return it; otherwise, return the coloring.
\end{proof}

\section{Combining height and number of leaves}

In any forest, the product of the number of leaves and the number of vertices in the longest root-to-leaf path is at least $n$; this inequality is tight for forests formed by $k$ paths of length $n/k$. By combining this inequality with the methods from the previous two sections, we obtain additional paired approximation results. For undirected graphs the results obtained in this fashion are relatively weak, but we may also apply the same technique to directed graphs obtaining more interesting results.

\begin{theorem}
In linear time, given an $n$-vertex directed graph $G$, it is possible to find a path $P$ in the graph, and a subset $A$ of vertices in the graph such that the induced subgraph $G[A]$ is acyclic, such that $|P|\cdot|A|\ge n$.
\end{theorem}

\begin{proof}
Perform a depth-first search in $G$, let $P$ be the longest root-to-leaf path in the depth-first search forest, and let $A$ be the set of leaves of the depth-first search forest.
\end{proof}

\begin{corollary}
In any $n$-vertex directed graph $G$, if $P^*$ is the vertex set of a longest directed path and $A^*$ is a largest vertex set such that $G[A^*]$ is acyclic, then $|P^*|\cdot|A^*|\ge n$.
\end{corollary}

Again, a graph in the form of a disjoint union of complete graphs shows that this tradeoff is tight.

\begin{theorem}
In linear time, given an $n$-vertex directed graph $G$ and an input parameter $0<\epsilon<1$, it is possible to find either a path that approximates the longest path to within a factor of $n^{\epsilon}$, or a set $A$ that induces an acyclic subgraph and approximates the largest vertex set of an acyclic induced subgraph to within a factor of $n^{1-\epsilon}$.
\end{theorem}

\begin{proof}
Compute $P$ and $A$ as above. If the path contains at least $n^{1-\epsilon}$ vertices, return it; otherwise, return the acyclic subset.
\end{proof}

Unlike our results on undirected longest path problems, this result combines two problems that can both be proved to be individually hard to approximate better than any $O(n^{1-\epsilon})$ factor.  For directed longest paths, this is the result of~\cite{BjoHusKha-ICALP-04}; we are not aware of past work on approximability of acyclic induced subgraphs (although they are complementary to the more well-studied feedback vertex sets) but there is an easy approximation-preserving reduction from undirected maximum independent set to acyclic induced subgraph: replace every edge of the given undirected graph by a cycle of two directed edges.

A similar approximation result also applies when combining the asymmetric $(1,2)$-TSP and maximum acyclic induced subgraph problems: if an asymmetric distance function is defined by a directed graph $G$, with distance 1 for graph edges and distance two for non-edges, then the leaves of a depth-first search forest for $G$ form a set $A$ that induces an acyclic subgraph, and a preorder traversal of the forest forms a tour with length $n+|A|$. If $|A|\ge\epsilon n$, the result is a $1/\epsilon$-approximation to the maximum acyclic induced subgraph, and otherwise the result is a $(1+\epsilon)$-approximation to the TSP.

\section{The Hadwiger conjecture}

Although the Hadwiger conjecture remains unproven, a paired approximation resembling but weaker than the one that would be implied by (a polynomial-time algorithm for) the conjecture can be obtained by combining known algorithms.

\begin{theorem}
In any graph $G$, one can find in polynomial time a coloring with $k$ colors and a clique minor of $G$ with $\Omega(k/\sqrt{\log k})$ vertices.
\end{theorem}

\begin{proof}
Apply a greedy coloring algorithm to $G$, by removing the vertex with minimum degree, recursively coloring the remaining graph, restoring the removed vertex, and giving it the lowest-numbered color that is different from all its neighbors' colors. Let $d$ be the largest minimum degree of any subgraph $S$ formed during the removal process of the greedy algorithm; then by a known result of Kostochka~\cite{Kos-Comb-84}, a clique minor with $\Omega(d/\sqrt{\log d})$ vertices may be found in $S$ in polynomial time. The number $k$ of colors used by the coloring is at most $d+1$, so the result follows.
\end{proof}

\begin{corollary}
In polynomial time, given an $n$-vertex directed graph $G$ and an input parameter $0<\epsilon<1$, it is possible to find either a coloring that approximates the chromatic number to within a factor of $O(n^{\epsilon}\log^{1/4} n)$, or a clique minor that approximates the Hadwiger number to within a factor of $O(n^{1-\epsilon}\log^{1/4} n)$.
\end{corollary}

\section{Set cover and hitting set}

To prove approximation hardness for the paired problem of set cover and hitting set, it will be convenient to think of both problems as being defined by a bipartite graph $G=(U,V,E)$ where $U$ has one vertex per set in the input set family, $V$ has one vertex per element in the input set family, and $E$ specifies the containment relation between sets and vertices. In this formulation, a set cover is a subset of $U$ that contains at least one member adjacent to every vertex in $V$, while a hitting set is a subset of $V$ that contains at least one member adjacent to every vertex in $U$. We define the \emph{transpose} $G^T=(V,U,E)$; a set cover in $G$ is a hitting set in the transpose and vice versa.

\begin{figure}[t]
\centering\includegraphics[width=4in]{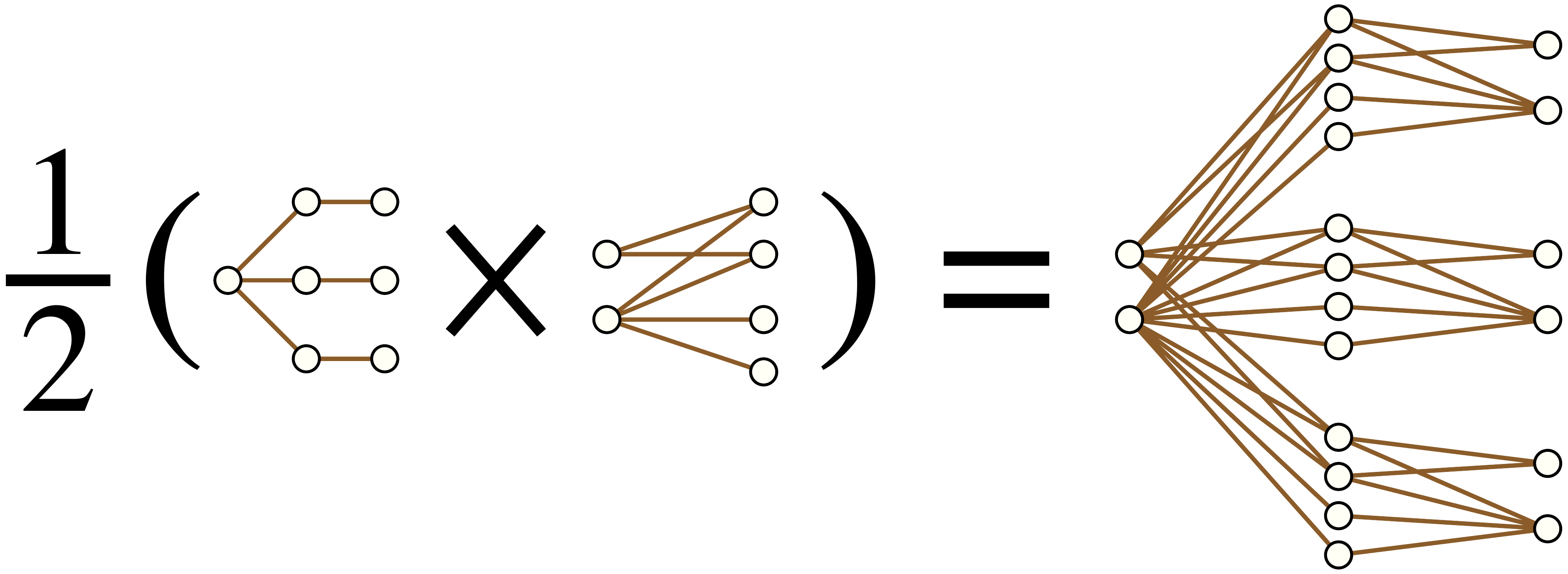}
\caption{The construction of $kG$ (here $k=3$) as one component of a tensor product of the tree $S_k$ and $G$.}
\label{fig:tensor}
\end{figure}

Given a bipartite graph $G$ and an integer $k$, we define a new bipartite graph $kG$ as follows.
$kG$ contains $k+1$ vertices $(u_i,j)$ ($0\le j\le k$) for each vertex $u_i\in U$, and $k$ vertices $(v_i,j)$ ($1\le j\le k$) for each vertex $v_i$ in $V$. If $G$ contains an edge from $u_i$ to $v_{i'}$, then $kG$ contains edges from $(u_i,0)$ to $(v_{i'},j)$ and from $(u_i,j)$ to $(v_{i'},j)$ for every $j$ in the range $1\le j\le k$.

$kG$ may alternatively be constructed using tensor products of graphs. The tensor product of any two bipartite graphs is disconnected, with two connected components. $kG$ is one of the components of the tensor product $S_k\times G$, where $S_k$ is the tree formed by subdividing each edge of a complete bipartite graph $K_{1,k}$ (Figure~\ref{fig:tensor}).

Let $\mathop{\mathrm{Cover}}(G)$ denote the size of the optimal set cover of $G$, and $\mathop{\mathrm{Hit}}(G)=\mathop{\mathrm{Cover}}(G^T)$ denote the size of the optimal hitting set of $G$.

\begin{lemma}
For any integer $k>0$ and bipartite graph $G$,
$\mathop{\mathrm{Cover}}(kG)=\mathop{\mathrm{Cover}}(G)$ and
$\mathop{\mathrm{Hit}}(kG)=k\cdot\mathop{\mathrm{Hit}}(G)$.
\end{lemma}

\begin{proof}
In any covering set for $kG$, we may assume without loss of generality that all vertices of the covering set have the form $(u_i,0)$, for any $(u_i,j)$ with $j\ne 0$ covers a strict subset of the vertices covered by $(u_i,0)$. Thus, it forms a covering set in $G$, and any covering set in $G$ can be transformed into a covering set in $kG$ in the same way. For the hitting set problem, hitting each subgraph $(U,j)$ with $j\ne 0$ requires forming a hitting set in each $(V,j)$, and if such a set is included in the hitting set then $(U,0)$ will also be hit. Therefore, an optimal hitting set in $kG$ consists of $k$ optimal hitting sets in $G$, one in each subset $(V,j)$.
\end{proof}

Our reduction will use a graph $k^*G$ formed as the disjoint union of $kG^T$ and $(kG^T)^T$.

\begin{lemma}
\label{lem:starsize}
For any integer $k>0$ and bipartite graph $G$,
$\mathop{\mathrm{Cover}}(k^*G)=\mathop{\mathrm{Hit}}(k^*G)=k\cdot\mathop{\mathrm{Cover}}(G)+\mathop{\mathrm{Hit}}(G)$.
\end{lemma}

\begin{proof}
This follows immediately from the previous lemma.
\end{proof}

Thus, at a cost of expanding the graph size by a factor of $2k$, we may amplify the set cover number relative to the hitting set number, and symmetrize the problem so that a paired approximation to the symmetrized problem gives the same information no matter whether an approximation returns a set cover or a hitting set for its given instance.
We define a \emph{reduced solution} to the set cover problem in $kG$ to be a solution in which only vertices of  the form $(u_i,0)$ appear in the cover. We define a reduced solution to the set cover or hitting set problem in $k^*G$ similarly: whenever it is possible for the solution to use a vertex $(u_i,0)$ in place of a vertex of the form $(u,i,j)$ for $j\ne 0$, the vertex $(u_i,j)$ must not be included. It is straightforward to transform any solution into a reduced solution of equal or smaller size in polynomial time.

\begin{theorem}
For all $\epsilon>0$,
unless NP${}\subset\mathop{\mathrm{DTIME}}[n^{O(\log\log n)}]$, it is not possible to solve the paired approximation problem of set cover and hitting set to within an approximation ratio of $(1-\epsilon)\log n$, where $n$ denotes the total size of the input set cover instance.
\end{theorem}

\begin{proof}
Feige~\cite{Fei-JACM-98} describes a reduction from an NP-hard problem to the set cover problem, such that an answer to the set cover problem that is accurate to an approximation ratio of $(1-\epsilon)\log N$ would allow one to infer a correct answer to the original problem. Here $N$ denotes the number of elements in the resulting set cover instance; if the satisfiability problem instance has size $s$, then the running time of the reduction is $O(s^{O(\log\log s)})$. From this he infers that, unless NP${}\subset\mathop{\mathrm{DTIME}}[n^{O(\log\log n)}]$, it is not possible to solve the set cover problem to within an approximation ratio of $(1-\epsilon)\log N$.

In the bipartite graph $G=(U,V,E)$ representing the set cover instance produced by Feige's reduction, $\log |U|=O(\epsilon\log|V|)$. That is, the number of sets in the set family is significantly smaller than the number of elements in their union. The elements in the set family's union consist of pairs of values that specify a ``partition problem point'' and a ``random string''; the random string specifies a sequence of positions in which to probe a probabilistically checkable proof for the satisfiability instance, but the number of partition problem points ($m$, in Feige's notation) is chosen to be a much larger number, so that the logarithm of the number of elements is $(1-\epsilon/2)\log m$. The sets in the set family, on the other hand, may be specified by triples of values that specify the identity of a prover for the probabilistically checkable proof, a sequence of positions in which to probe the proof, and the answers to be found at those positions. The number of sets is therefore roughly comparable to the number of possible random strings, which is much smaller than $m$.

To complete the proof, we transform $G$ to the graph $k^*G$, where $k=|U|+1$. An optimal solution to either the set cover instance or the hitting set instance in $k^*G$, with total size $q$, may be transformed to an optimal solution to the set cover problem for $G$ with total size $\lfloor q/k\rfloor$, because $k$ is chosen to be so large that division by $k$ reduces the size formula from Lemma~\ref{lem:starsize} to the size of the set cover problem plus a number less than one. Similarly, any reduced solution with total size $q'$ for $k^*G$ may be transformed to a solution to the set cover problem for $G$ with total size $\lfloor q'/k\rfloor$. Therefore, it is hard to approximate the paired problem for $k^*G$ to within a factor better than $(1-\epsilon)\log N= (1-O(\epsilon))\log(k^2N)=(1-O(\epsilon))\log|k^*G|$.
\end{proof}

\section{Clique and independent set}

Our proof of hardness for the paired clique and independent set problem will involve modifying a standard reduction proof of hardness for approximating the maximum clique, by adding edges to the graph resulting from this reduction so that it loses all of its large independent sets without gaining any new large cliques. The property we need of the set of edges to be added can be encapsulated in the form of a Ramsey-theoretic property: define an $n$-vertex graph $G$ to be an \emph{$f(n)$-(biclique,independent) Ramsey graph} if $G$ contains neither an $f(n)$-vertex independent set nor a complete bipartite subgraph $K_{f(n),f(n)}$. Standard techniques for random graphs show that a graph chosen uniformly among all $n$-vertex graphs (equivalently, with each edge included independently at random with probability $1/2$) is with high probability an $O(\log n)$-(biclique,independent) Ramsey graph but we are willing to accept a weaker bound in order to avoid the use of randomness.

\begin{lemma}
\label{lem:Ramsey}
For any constant $\epsilon$,
there exists a deterministic algorithm that takes a parameter $n$ and produces in polynomial time an $n$-vertex $O(n^\epsilon)$-(biclique,independent) Ramsey graph.
\end{lemma}

\begin{proof}
Our construction uses the \emph{bipartite Ramsey graph} construction of \cite{BarKinSha-STOC-05}: this  construction produces, for any $n$ and any constant $\delta>0$ a bipartite graph $(U,V,E)$ in which, for every two subsets of $U$ and $V$ of size $n^\delta$ there is at least one edge connecting $U$ to $V$ and at least one edge missing between $U$ and $V$, so that these two subsets induce neither a complete bipartite graph nor its complement.

We may assume without loss of generality that $n=2^{k}$ for some $k$ (for other values of $n$, one may round $n$ up to the next larger power of two, and then take an arbitrary $n$-vertex induced subgraph of the resulting larger graph). Label the vertices of the graph by distinct length-$k$ bitstrings. For each $i=0,1,\ldots k-1$, we form a bipartite Ramsey graph for $\delta=\epsilon/2$ connecting the vertices in which bit $i$ is $0$ to the vertices in which bit $i$ is 1; we form our overall graph $G$ by connecting two vertices $u$ and $v$ by an edge if they are connected by an edge in $G_i$, where $i$ is the first position at which the labels of $u$ and $v$ differ.

We now show that $G$ cannot contain a large biclique. Suppose $A$ and $B$ are two subsets of size $\Omega(n^\epsilon)$; we must show that some potential edge from $A$ to $B$ is not present. We describe below an algorithm that finds this missing edge, by searching for a number $i$, and large subsets $A'\subset A$ and $B'\subset B$, with the following properties:
\begin{enumerate}
\item The vertices in $A'\cup B'$ all have equal labels up to but not including their $i$th bits
\item The $i$th bits of the labels in $A'$ are all equal,
\item The $i$th bits of the labels in $B'$ are all equal, and
\item The $i$th bits of the labels in $A'$ differ from the $i$th bits of the labels in $B'$.
\end{enumerate}
Once these properties are met, the bipartite Ramsey graph property of $G_i$ will guarantee a missing edge.

To find $i$, test each number $0$, $1$, $2$, etc., in sequence. Prior to testing each value of $i$, we will guarantee that the remaining vertices in $A$ and $B$ all have equal labels up through bit $i-1$; after testing $i$, we will either find a large set of vertices with differing labels or we will extend this guarantee, showing that the remaining vertices all have equal labels up through bit $i$.
Specifically, if we are testing position $i$, and the majority value of the bit in position $i$ of the labels of the remaining vertices in $A$ differs from the majority value for $B$, return the vertices having these differing majority values as the large sets $A'$ and $B'$. If the majority of the labels in $A$ and in $B$ have the same value for bit $i$, but either of these sets contain a minority of at least $|A|/2k$ or $|B|/2k$ vertices with the other bit value, then return that minority as one of $A'$ or $B'$ and the majority on the other side as the other set. Finally, if the number of vertices and $A$ and $B$ having a non-majority value in position $i$ is less than  $|A|/2k$ and  $|B|/2k$ respectively, remove all those non-majority vertices from $A$ and $B$ and continue to $i+1$.
The number of removed vertices, over the course of testing all values of $i$, is at most half the starting number of vertices. Therefore, it is not possible to continue past $i=\log n$, because to do so would imply that half the vertices in the original sets $A$ and $B$ have equal labels, violating the assumption that all labels are distinct. Thus, for some $i$ we find sets $A'$ and $B'$ of sizes at least $\Omega(n^\epsilon/\log n)=\Omega(n^\delta)$ such that the edges connecting $A'$ to $B'$ are all drawn from $G_i$; the fact that there is a missing edge from $A'$ to $B'$ then follows from the bipartite Ramsey property.

The bipartite Ramsey property definition, and our construction of $G$ from bipartite Ramsey graphs, are both self-complementary. Therefore the same proof as above shows that $G$ can have no two large subsets $A$ and $B$ in which all edges in $A\times B$ are missing. A fortiori, it can also have no single large independent set, because any two halves of an independent set would have no edges from one half to the other.

We have shown that every two large subsets of $G$ neither form a complete bipartite subgraph nor an independent set, fulfilling the definition of a (biclique,independent) Ramsey graph.
\end{proof}

We use these Ramsey graphs to perturb the graphs formed by Zuckerman's~\cite{Zuc-STOC-06} inapproximability reduction for maximum clique; the perturbation will remove any large independent sets that these graphs may have, but we will also need to show that it does not introduce new large cliques. To do so, we must examine in more detail the clique reduction.
Zuckerman shows that there is an NP-complete problem, and a system of probabilistically checkable proofs and proof checkers, with the following properties. A positive problem instance of length $N$ corresponds to a proof represented as a bitstring of length $\sigma=N^{O(1)}$, while a negative problem instance has no valid proof. A proof checker is a probabilistic Turing machine that generates $R$ truly random bits (for a parameter $R=\Theta(\log n)$ and uses these bits, the input instance, and past probe results to determine a sequence of  bit positions at which to probe the proof by examining the bit at that position. At any point, it may accept or reject the proof based only on the positions it has examined. Among the positions probed by the checker, $\epsilon R$ of the positions are \emph{free bits}, positions with the property that the checker may continue without rejecting or accepting no matter what value is seen there; in the remaining proof positions, the checker expects to see a certain bit and will reject if it does not see it. A valid proof will always be accepted; an invalid proof will be rejected with probability $1-1/2^{(1-\epsilon)R}$.

The reduction from this system to finding cliques in graphs is the same one used by~\cite{FeiGolLov-JACM-96} and many subsequent papers on hardness of approximation, as follows. Given an instance of the starting NP-complete problem, construct from it a graph the vertices of which represent all possible accepting runs of a prover: that is, for each possible sequence of random bits, and for each possible sequence of results from probing free bits, simulate what the prover would do for those bits. In the simulation, whenever a non-free bit is probed, return the result that will not make the prover reject. If the simulation eventually accepts, construct a vertex in the graph. Link two vertices by an edge if their simulations returned the same bit for every probe position in the proof that they both examined.
The graph has at most $2^{(1+\epsilon)R}$ vertices (at most one for each combination of random and free bits). Any clique in this graph corresponds to a set of prover simulations that agree on any bits of the proof that they probed in common. That is, a clique may be represented as a string of $\sigma$ bits, describing a purported proof, and the vertices in the clique are the simulations that accepted after probing that string. If there is a valid proof, the simulations that probe this proof for any of the $2^R$ choices of random bits form a clique with $2^R$ vertices. But if a string of $\sigma$ bits represents an invalid proof, then most of the simulations that probe that string reject it, and its clique will have size only $2^{\epsilon R}$.

Now let $G$ be a graph formed by this reduction. Assume without loss of generality that $R$ is chosen sufficiently large that $2^{\epsilon R}\ge 2\sigma$ (if not, use modified provers that perform $O(1)$ runs of the original prover and accept only if all of these runs accept). Let $H$ be a $2^{\epsilon R}$-(biclique,independent) Ramsey graph on the same set of vertices. What can we say about the clique and independence numbers of $G\cup H$? First, clearly, $G\cup H$ has independence number at most $2^{\epsilon R}$, because any independent set in $G\cup H$ is also an independent set in $H$. But second, as we now show, if $G$ has no large cliques then neither does $G\cup H$.

\begin{lemma}
\label{lem:noclique}
With the notation as above, if $G$ comes from a negative instance to the starting NP-complete problem (so that it has no cliques of size $2^{\epsilon R}$) then $G\cup H$ has no cliques of size $2^{2\epsilon R}$.
\end{lemma}

\begin{proof}
Let $S$ be a set of vertices in $G$ (that is, simulations of proof checkers) of cardinality $2^{2\epsilon R}$; we must show that $S$ is not a clique.
Let $P$ be a bitstring of length $\sigma$, where the value of the bit at each position of $P$ is the value seen at that position by a majority of the simulations in $S$ that examined that position. It cannot be the case that the majority of simulations in $S$ agree with $P$ in each of their probes, for otherwise this majority would form a clique of size $2^{2\epsilon R-1}$ in $G$, contradicting the assumption that there is no such clique. Therefore,  some subset of $2^{2\epsilon R-1}$ vertices in $S$ disagree with $P$ in at least one position. By the pigeonhole principle and the assumption that $2^{\epsilon R}\ge 2\sigma$ there exist a position $i$ in $P$ and a subset $A$ of $2^{\epsilon R}$ vertices in $S$ that disagree with $P$ in position $i$. Let $B$ be the simulations in $S$ that probed position $i$ and saw the majority value there. No edges in $G$ connect $A$ to $B$ (each pair of a vertex in $A$ and a vertex in $B$ represents two simulations that disagree on position $i$, and an edge only exists between simulations that agree on all their common probe positions). Additionally, $H$ does not contain all the possible edges from $A$ to $B$: $A$ is a large set and $B$, representing as it does the majority view for the value of position~$i$, is at least as large, so the set of all possible edges would form a biclique that is too large to exist in $H$. Therefore, there is at least one pair of a vertex in $A$ and a vertex in $B$ that is connected neither in $G$ nor in $H$; this missing edge in $G\cup H$ shows that $S$ cannot be a clique.
\end{proof}

\begin{theorem}
For all $\epsilon>0$,
unless \PNP, it is not possible to solve the paired approximation problem of clique and independent set to within an approximation ratio of $n^{1-\epsilon}$, where $n$ denotes the number of vertices of the input graph.
\end{theorem}

\begin{proof}
If an approximation algorithm did exist, we could use it to solve SAT in polynomial time, as follows: given an instance of SAT, use the reduction of~\cite{Zuc-STOC-06} to construct a graph $G$ with $n$ vertices (where $n$ is a polynomial in the input instance size) that has a clique of size $n^{1-\epsilon}$ if the input SAT instance is positive and no cliques larger than $n^\epsilon$ otherwise. Let $H$ be a (biclique,independent) Ramsey graph on $n$ vertices, and let $F$ be the $2n$-vertex disjoint union of $G\cup H$ and its complement. If the starting instance of SAT is positive, $F$ has clique number and independence number both $\Theta(n^{1-\epsilon})$, while if the starting instance of SAT is negative, $F$ has (by Lemma~\ref{lem:noclique}) clique number and independence number both $O(n^{2\epsilon})$. Thus, an approximation to either the clique number or the independence number in $F$ that is accurate to within a factor of $n^{1-\epsilon}/n^{2\epsilon}$ would allow us to determine whether $F$ is a positive instance. Substituting $\epsilon/3$ for $\epsilon$ gives the result.
\end{proof}

\section{TSP and MaxTSP}

\begin{theorem}
There exists a constant $\epsilon>0$ such that, unless \PNP, it is not possible to solve the paired approximation problem of $(1,2)$-TSP and $(1,2)$-MaxTSP to within an approximation ratio of $1+\epsilon$.
\end{theorem}

\begin{proof}
We use an approximation-preserving reduction from $(1,2)$-TSP; therefore, let $G$ be an $n$-vertex graph defining a $(1,2)$-TSP instance. From $G$ we define a $(3n-1)$-vertex graph $H$ from the union of $G$, the complement of $G$, and an $(n-1)$-vertex clique; we add edges connecting every vertex in the clique to every vertex in the complement of $G$, and do not add any edges from $G$ itself to the other two components.

Suppose that the $(1,2)$-TSP defined from $G$ has an optimal solution with length $n+x$ (and, for simplicity of exposition, assume $x>0$). Then the optimal $(1,2)$-TSP in $H$ has length $3n+x-1$: it consists of an optimal path in $G$ of length $n+x-2$, a Hamiltonian path through the other two components of length $2n-3$, and two edges of total length $4$ connecting these two paths. The optimal $(1,2)$-MaxTSP in $H$ has length $6n-x-5$: it consists of the path in the complement of $G$ that would be optimal as a $(1,2)$-TSP in $G$ and that has length $2n-x-1$, a path through $G$ and the clique that uses no edges of $H$ and has total length $4n-6$, and two edges of total length $2$ connecting these two paths.

Thus, if we could approximate the paired problem within $1+\epsilon$ in polynomial time, we could approximate $(1,2)$-TSP within $1+6\epsilon$ in the same time. Since $(1,2)$-TSP is not approximable better than $1+\frac{1}{740}$ unless \PNP~\cite{EngKar-ICALP-01}, the paired problem is not approximable better than $1+\frac{1}{4440}$.
\end{proof}

\raggedright
\bibliographystyle{abbrv}
\bibliography{paired}
\end{document}